\documentclass[12pt]{article}
\usepackage{amsfonts}
\usepackage{amsmath}
\usepackage{amssymb}
\usepackage{amsthm}
\usepackage{setspace}
\usepackage{fullpage}
\usepackage{graphicx}
\usepackage{url}
\usepackage{verbatim}
\usepackage{enumerate}

\usepackage{tikz}
\usetikzlibrary{arrows}

  \newcommand{\tikzmath}[2][]
     {\vcenter{\hbox{\begin{tikzpicture}[#1]#2
                     \end{tikzpicture}}}
     }

\def\commentout#1{}

\newcommand{\nc}[2]{\newcommand{#1}{#2}}
\nc{\A}{\mathsf{A}}
\nc{\C}{\mathbb{C}}
\nc{\R}{\mathbb{R}}
\nc{\Q}{\mathbb{Q}}
\nc{\Z}{\mathbb{Z}}
\nc{\N}{\mathbb{N}}
\nc{\cA}{\mathcal{A}}
\nc{\cB}{\mathcal{B}}

\nc{\g}{\mathfrak{g}}

\newcommand{\DeclareMyOperator}[1]{ \expandafter\DeclareMathOperator\csname #1\endcsname{#1} }

\def\Rep{\mathrm{Rep}}

\def\Vect{\mathrm{Vect}}

\newtheorem{theorem}{Theorem}
\newtheorem*{theorem*}{Theorem}

\newtheorem{lemma}[theorem]{Lemma}
\newtheorem*{tech-lemma}{Technical lemma}
\newtheorem{proposition}[theorem]{Proposition}

\newtheorem{corollary}[theorem]{Corollary}
\newtheorem{definition}[theorem]{Definition}

\newtheorem{conjecture}[theorem]{Conjecture}
\newtheorem*{conjecture*}{Conjecture}

\newtheorem*{question*}{Question}

\theoremstyle{remark}

\newtheorem*{remark*}{Remark}

\newtheorem*{example*}{Example}

\begin{document}
\title{The classification of chiral WZW models\\ by $H^4_+(BG,\Z)$\vspace{-.1cm}}
\author{Andr\'e Henriques}
\date{}
\maketitle

\vspace{-.3cm}
\abstract{
We axiomatize the defining properties of chiral WZW models.
We show that such models are in almost bijective correspondence with pairs $(G,k)$, where $G$ is a connected Lie group
and $k \in H^4_+(BG,\Z)$ is a degree four cohomology class subject to a certain positivity condition.
We find a couple extra models which satisfy all the defining properties of chiral WZW models, but which don't come from pairs $(G,k)$ as above.
The simplest such model is the simple current extension of the affine VOA $E_8 \times E_8$ at level
$(2,2)$ by the group $\Z_2$.

\tableofcontents


\section{Introduction}

\subsection{Simply connected WZW models}\label{sec: 1.1}

Associated to every simple, simply connected, compact Lie group $G$, and to every level $k\in\mathbb N:=\Z_{>0}$, are certain well-known unitary chiral conformal field theories ($\chi$CFTs), known as the \emph{chiral Wess-Zumino-Witten models}.
They where first introduced as field theories in their own right (as opposed to chiral `halves' of full CFTs) by Witten \cite{MR1151251}, 
following Spiegelglas \cite{MR1214323}, and 
were subsequently studied by many authors.

These $\chi$CFTs are usually described by means of a \emph{construction}.
It is interesting to note that they can also be defined in terms of \emph{characterizing properties}.
Using unitary vertex operator algebras (VOAs) as a mathematical formalism for $\chi$CFTs\footnote{Throughout this article, all VOAs will be assumed to be one-dimensional in degree zero.}, we propose:

\begin{definition}\label{def: s-c-chir WZW}
A ``simply connected chiral WZW model'' is a unitary VOA which is rational, and generated in degree 1.
\end{definition}

Here, we call a unitary VOA \emph{rational} if it has only finitely many isomorphism classes of irreducible unitary modules.\footnote{The term `rational VOA' has many meanings in the literature (see \cite[Appendix]{MR3339173} for an overview).
The present rather crude notion will be sufficient for our purposes.}\medskip

Let us show how the above definition agrees with the more standard approach
(see \cite[\S4]{Ai+Lin:Unitary-structures-of-VOAs} for an expanded version of the argument that we are about to present).
Let $V$ be a VOA as above, and let $J^a\in V_1$, $a=1,\ldots,n$, be a basis of its degree one part.
Then the constants $f^{ab}_c$ and $k^{ab}$ which appear in the OPE
\[
J^a(z)J^b(w) = \frac{k^{ab}}{(z-w)^2}+\frac{f^{ab}_c J^c(w)}{z-w}+{\rm reg.}
\]
endow $\g:=V_1$ with the structure of a Lie algebra, and with an invariant bilinear form $\kappa:\g\times\g\to \R$.
The universal affine VOA $V_{\g,\kappa}$ maps to $V$, and our assumption that the latter is generated in degree one is equivalent to $V$ being a quotient of $V_{\g,\kappa}$.

Now, if $V$ is unitary, $\g$ acquires a positive definite hermitian form in addition to the invariant bilinear form.
Those two pieces of data combine to a real structure of compact type, and so $\g$ is the direct sum of a semi-simple Lie algebra $\g^{ss}$ and an abelian Lie algebra $\mathfrak z$ \cite[Thm.\,4.10]{Ai+Lin:Unitary-structures-of-VOAs}:
\[
\g\,=\,\g^{ss}\oplus\mathfrak z\,=\,\g_1\oplus\g_2\oplus\ldots\oplus\g_n\oplus\mathfrak z.
\]
On each simple summand $\g_i$, the bilinear form $\kappa$ is a positive multiple of the basic inner product (no constraint imposed on $\kappa|_{\mathfrak z}$)
and $V$ is isomorphic to $L_{\g,\kappa}$, the quotient of $V_{\g,\kappa}$ by its unique maximal ideal.

If $V$ is furthermore assumed to be rational, then we also have $\mathfrak z=0$.
Simply connected chiral WZW model as defined above
are therefore in bijective correspondence with pairs $(\g,k)$ consisting of a semi-simple Lie algebra  $\g=\g_1\oplus \ldots \oplus \g_n$ (with real structure of compact type),
and an $n$-tuple of positive integers $k=(k_1,\ldots, k_n)\in\mathbb N^n$.
If one wants, one can rephrase the result by saying that simply connected chiral WZW model are classified by a pair $(G,k)$ where $G$ is a compact simply connected Lie group,
and $k$ is an element in the positive part
$H^4_+(BG,\Z)\cong\N^n$
of $H^4(BG,\Z)\cong\Z^n$.

We summarize the conclusion of the above discussion in the following theorem:

\begin{theorem}
There is a natural bijective correspondence between simply connected chiral WZW model and pairs $(G,k)$, where $G$ is a compact simply connected Lie group and $k$ is an element of $H^4_+(BG,\Z)$.
\end{theorem}

The goal of this paper is to generalize the above result to the case of non-simply connected groups.

\subsection{Non-simply connected WZW models}\label{intro2}

Wess-Zumino-Witten models for non-simply connected Lie groups have been considered by many authors, and so have their chiral halves \cite{MR946997}\cite[\S2]{MR992362}\cite[\S6]{MR1409292}.
In the mathematical literature, this class of models was studied in \cite{MR1408523}\cite{MR1822111}.
Their representations were classified, and their fusion rules computed.

The cohomology group $H^4(BG,\Z)$ is briefly mentioned in \cite[Appendix]{MR1151251}, and its role in setting up certain aspects of the chiral WZW models was explained in \cite[\S5]{MR1291698}.
But the question of \emph{classification} of chiral WZW models has, to our knowledge, never been formulated in the way we do it here.

In order to formulate our question, we first need a definition of chiral WZW models, analogous to the one given in the previous section.
In Definition~\ref{def: g-chir WZW}, we present a preliminary notion of a not necessarily simply connected chiral WZW model
(in Section~\ref{sec: LG}, we will justify it in an informal way, from the point of view of geometric quantization).
Later, we will fine-tune our definition (Definition~\ref{def: chir WZW}) so as to make our main theorem (Theorem \ref{thm: Classific WZW}) be true.\medskip

An extension of VOAs $W\subset V$ is a called a \emph{simple current extension} if $V$ decomposes as a direct sum $V=\bigoplus_{\lambda\in \pi}M_\lambda$ of invertible $W$-modules
indexed by a (possibly infinite) abelian group $\pi\subset \mathrm{Rep}(W)$, and the VOA structure on $V$ is compatible with the $\pi$-grading.

\begin{definition}\label{def: g-chir WZW}
A ``general chiral WZW model'' is a unitary rational VOA $V$ which is a simple current extension of the sub-VOA generated by its degree one part.
\end{definition}

Given a compact Lie group $G$, the positive part of $H^4(BG,\Z)$ is the subset of elements whose image under the Chern--Weil homomorphism $$H^4(BG,\Z)\to \mathrm{Sym}^2(\g^*)$$ are positive definite bilinear forms
on the Lie algebra of $G$.
We denote it by
\[
H^4_+(BG,\Z)\subset H^4(BG,\Z).
\]
We refer the reader to the end of Section \ref{sec : H^4} for a more precise description of $H^4_+(BG,\Z)$.
The result which we had hoped to be true is the following:\medskip

\noindent
\emph{``There is a natural bijective correspondence between general chiral WZW models and pairs $(G,k)$ where $G$ is a compact connected Lie group (not necessarily simply connected),
and $k$ is an element of $H^4_+(BG,\Z)$.''}\medskip

Unfortunately, with Definition \ref{def: g-chir WZW},
the natural map
\begin{equation}\label{eq: not an iso :-(}
\{(G,k)\,|\,
G \text{ connected, } k\in H^4_+(BG,\Z)\}
\,\,\to\,\,\,
\{\text{General chiral WZW models}\}
\end{equation}
is neither injective nor surjective.
The simplest example that illustrates the lack of injectivity is $SU(2)$ level $1$, versus $U(1)$ level $1$:\footnote{
A better name is to call this ``$U(1)$ level $2$'', and to reserve ``$U(1)$ level $1$'' for the free fermion super-VOA.
We called it ``$U(1)$ level $1$'' because it corresponds to the generator of $H^4(BU(1),\mathbb Z)$.
}
those two models yield isomorphic VOAs \cite{MR0626704}.
In order to restore injectivity, we need to add the Lie algebra of $G$ as part of our data, which leads us to the following modified definition (still preliminary):\medskip

\addtocounter{theorem}{-1}
\begin{definition}
A\put(-16.5,0){$'$}\;\! ``general chiral WZW model'' is a pair $(V,\g)$ consisting of a unitary rational VOA $V$ and a Lie algebra $\g\subset V_1$, such that
$V$ is a simple current extension of the sub-VOA generated by $\g$.
\end{definition}


Even after this modification,
we still have the problem that the map (\ref{eq: not an iso :-(}) is not surjective:
the pairs $(G,k)$ only parametrize a proper subset of the general chiral WZW models.
The problem can be entirely blamed on the existence of the affine VOA $L_{E_8,2}$ 
associated to the Lie group $E_8$ at level $2$.
For all simply connected chiral WZW models $L_{\g,k}$
\emph{except the one associated to $E_8$ at level $2$},
the set of invertible $L_{\g,k}$-modules is in bijection with the centre of $G$ \cite{MR1096120}.
The VOA $L_{E_8,2}$, on the other side, admits a non-trivial invertible module, despite the fact that $Z(E_8)=\{e\}$.

In writing this paper we were faced with the following unpleasant choice.
We could either stick to a cleaner notion of chiral WZW model, such as the ones presented in Definitions \ref{def: g-chir WZW} or $\ref{def: g-chir WZW}'$,
at the cost of making our main theorem inelegant, or add a weird-looking clause to our definition so as to ensure that our main theorem looks good. We chose the second option: 
our working definition of chiral WZW model (Definition \ref{def: chir WZW})
explicitly excludes the counterexamples that would otherwise be there.
As a consequence, our main theorem (Theorem~\ref{thm: Classific WZW}) does not have to mention these counterexamples.

Our choice was also guided by our conviction that the VOAs built using the non-trivial invertible module of $L_{E_8,2}$ should not be called ``chiral WZW models''.

Let us write $M_{E_8,2}$ for the non-trivial invertible $L_{E_8,2}$-module.
We modify Definition~$\ref{def: g-chir WZW}'$
so as to explicitly rule out the possibility that $M_{E_8,2}$ gets used in the simple current extension.
Let us call an extension $W\subset V$ \emph{$E_{8,2}$-contaminated} if one can write $W$ as a tensor product $W=W'\otimes L_{E_8,2}$
and find a sub-$W$-module $M\subset V$ of the form $M=M'\otimes M_{E_8,2}$ for some $W'$-module $M'$:

\begin{definition}\label{def: chir WZW}
A ``chiral WZW model'' is a pair $(V,\g)$ consisting of a rational unitary VOA $V$ and a Lie algebra $\g\subset V_1$, such that $V$ is a simple current extension of the sub-VOA generated by $\g$, and the extension is not $E_{8,2}$-contaminated.
\end{definition}

With this definition in place, we can state our main theorem:

\begin{theorem}\label{thm: Classific WZW}
There is a natural bijective correspondence between chiral WZW models and pairs $(G,k)$ where $G$ is a compact connected (not necessarily simply connected) Lie group and $k$ is an element of $H^4_+(BG,\Z)$.
\end{theorem}

Given a Lie group $G$ as above with Lie algebra $\mathfrak g$ and a positive level $k\in H^4_+(BG,\mathbb Z)$, 
the corresponding VOA is given by
\begin{equation}\label{my VOA}
V_{G,k}
:=\pi\ltimes (L_{\g^{ss},k}\otimes M_{\mathfrak z}).
\end{equation} 
It is a simple current extension of the tensor product $L_{\g^{ss},k}\otimes M_{\mathfrak z}$ of an affine VOA (associated to the semi-simple part of $\g$) and a Heisenberg VOA (associated to the center of $\mathfrak g$) by the abelian group $\pi:=\pi_1(G)$.

\begin{remark*}
The proof of Theorem \ref{thm: Classific WZW} is conditional on a conjecture (Conjecture~\ref{conj: unitary simple current ext}), according to which 
simple current extensions of unitary VOAs by unitary modules are again unitary.
The construction of $V_{G,k}$ does not rely on Conjecture~\ref{conj: unitary simple current ext}, but its unitarity does.
\end{remark*}

The interest of our theorem is that the group $H^4_+(BG,\Z)$ is never mentioned in Definition~\ref{def: chir WZW}. It just comes out from the proof.

Note that our result is in stark contrast with the classification of full WZW models by (positive) elements of $H^3(G,\Z)$ \cite{MR1048699}
(as a side remark, we note that, unlike Theorem~\ref{thm: Classific WZW}, the latter is only valid when $G$ is semi-simple; sigma-models with target a torus are classified by translation-invariant Riemannian metrics, and those are not quantized).

In the last section of our paper, we describe the chiral WZW models in the formalism of chiral conformal nets.
We do not have an analog of Theorem~\ref{thm: Classific WZW} in that context.
Chiral conformal nets are conjecturally equivalent to unitary VOAs~\cite{arXiv:1503.01260}
(both are mathematical formalisations of the concept of a unitary $\chi$CFT), and
it is tempting to try to formulate a version of Theorem~\ref{thm: Classific WZW} for chiral conformal nets.
However, unlike for VOAs, the property of being generated in degree $1$ is not easily expressible in terms of conformal nets.
\bigskip\medskip

\noindent
{\large \bf Acknowledgments}\medskip\\
I am grateful to Marcel Bischoff, Christoph Schweigert and Constantin Teleman for useful comments, and to Scott Carnahan and Sebastiano Carpi for extensive feedback on an earlier version of this manuscript.

This research was supported by the Leverhulme trust, the EPSRC grant ``Quantum mathematics and computation'', and the ERC grant No 674978 under the European Union's Horizon 2020 research innovation programme.
\section{Geometric quantization}\label{sec: LG}

In this section, we provide an informal construction of the chiral WZW models from the point of view of loop groups and geometric quantization.
We then use this to justify our Definition \ref{def: g-chir WZW}.\medskip

Let $G$ be a compact connected Lie group, and let $k\in H^4_+(BG,\Z)$ be a level.
The construction goes as follows:
\begin{enumerate}
\item To begin with, for every circle $S$ (here, a `circle' is a smooth oriented manifold that is diffeomorphic to $S^1$) one can use the level $k$ to get a central extension\footnote{The central extension $\widetilde{LG}$ depends on $k$, even though this is not reflected in our notation.}
\[
1\,\to\, U(1)\,\to\, \widetilde{LG} \,\to\, LG \,\to\, 1
\]
of the loop group $LG:=\mathit{Map}(S,G)$.
See \cite[Chapt 4]{MR900587}
\cite{MR887993}
for the construction in the case of simply connected gauge groups, and 
\cite[Thm 3.7]{MR2610397} \cite[Thm A]{Waldorf-WIP}
for the general case.
This group satisfies the locality axiom of quantum field theory, in the sense that loops supported in disjoint intervals lift to commuting elements of $\widetilde{LG}$ \cite[\S3.3]{Waldorf-WIP} (\cite[Lem~3.1]{MR1231644} in the simply connected case).

\item Given the further data of a conformal disc $\mathbb D$ bounded by $S$, one constructs a certain representation $H_0=H_0(\mathbb D)$ of $\widetilde{LG}$, called the vacuum sector associated to $\mathbb D$.
It is the geometric quantisation of the Hamiltonian $LG$-space
\[
\begin{split}
LG/G = \{\text{moduli space of flat $G$-bundles over $\mathbb D$ trivialised over $\partial \mathbb D$}\}\\
= \{\text{holomorphic $G_\C$-bundles with a smooth trivialisation over $\partial \mathbb D$}\}
\end{split}
\]
with respect to the pre-quantum line bundle induced by the central extension.
This construction is described in \cite[Chapt 11]{MR900587} in the case of simple simply connected gauge groups, and
in \cite{MR2925299} in the case of tori.
It is in this step that it is important that $k$ be a positive element of $H^4(BG,\Z)$,
as this ensures that the K\"ahler structure on $LG/G$ is positive definite
and the pre-quantum line bundle has enough sections.

\item We now take $\mathbb D$ to be the standard unit disc. The circle group $S^1$ acts on $\mathbb D$ by rotations, and there is a corresponding action on $H_0$. The vertex algebra $V_{G,k}$ associated to $G$ and $k$ is the algebraic direct sum of the eigenspaces of this rotation action.
\end{enumerate}

\noindent
Unfortunately, it is not easy to describe the VOA structure on $V_{G,k}$ from the above perspective.
The construction also doesn't provide the projective action of $\mathrm{Diff}(S^1)$ on~$H_0$.
We will not try to address these shortcomings.
Instead, we just take the above heuristic construction as a road map for understanding the meaning of being a chiral WZW model.
We then use this to justify our Definition \ref{def: g-chir WZW}, at an informal level.
\bigskip

When $G$ is simply connected, then the above geometric quantization procedure yields the affine VOA $L_{\g,k}$. This is the construction in \cite[Chapt 11]{MR900587}.
(When $\g$ is not simple but only semi-simple,
we denote by $L_{\g,k}$ the tensor product of the affine VOAs associated to the simple summands of $\g$, with the understanding that $k$ is now an $n$-tuple of numbers)

Let us now take $G$ to be a Lie group that is not simply connected. If we assume that the fundamental group of $G$ is finite, then the universal cover $\tilde G$ is compact and we can proceed as follows.
Let $\tilde k\in H^4(B\tilde G,\Z)$ be the restriction of $k$.
The moduli space of flat $G$-bundles over $\mathbb D$ trivialised over $\partial \mathbb D$
decomposes as a disjoint union of $\pi_1(G)$ many copies of the corresponding moduli space for $\tilde G$.
The vacuum sector of $LG$ decomposes accordingly as $\pi_1(G)$ many representations of $L\tilde G$.
These representations are all invertible, and so we get that $V_{G,k}$ is a simple current extension of 
$L_{\g,k}$ by the group $\pi_1(G)$.

Remarkably, the above story goes through essentially unchanged when $\pi_1(G)$ is infinite (i.e., when $\g$ is reductive but not semi-simple).
In that case, we first consider finite covers $\tilde G\to G$ of $G$, so that $\tilde G$ is again a compact Lie group.
The story then proceeds as before:
the moduli space of flat $G$-bundles over $\mathbb D$ trivialised over $\partial \mathbb D$
is a disjoint union of copies of the corresponding moduli space for $\tilde G$, this time indexed over $Z:=\ker(\tilde G\to G)$.
The vacuum sector of $LG$ decomposes accordingly as a $Z$ many representations of $L\tilde G$,
and $V_{G,k}$ is a simple current extension of $V_{\tilde G,\tilde k}$ by $Z$.
Thus, for every finite quotient $Z$ of the fundamental group,
we get a direct sum decomposition of $V_{G,k}$ into $Z$ many summands, exhibiting it as a simple current extension of some smaller vertex algebra. 
Taking the common refinement of all these decompositions as $Z$ ranges over the finite quotients of $\pi_1(G)$, we get a direct sum decompositions of $V_{G,k}$ into $\pi_1(G)$ many summands.
The direct summand that corresponds to the trivial element in $\pi_1(G)$ is the tensor product of the affine VOA $L_{\g^{ss},\tilde k}$ with the Heisenberg VOA $M_{\mathfrak z}$, where
\[
\g=\g^{ss}\oplus \mathfrak z,
\]
is the decomposition of $\g$ into its semi-simple part $\g^{ss}$ and its center $\mathfrak z=\mathfrak z(\g)$.
At the end of the day, we get that $V_{G,k}$ is a simple current extension of 
$L_{\g^{ss},\tilde k}\otimes M_{\mathfrak z}$ by the infinite abelian group $\pi_1(G)$.

Now, the distinguishing feature of the vertex algebras $L_{\g^{ss},\tilde k}\otimes M_{\mathfrak z}$ is that they are unitary, and generated in degree one:
any unitary VOA that is generated in degree one is isomorphic to $L_{\g,k}\otimes M_{\mathfrak z}$ for some semi-simple Lie algebra $\g$, and a complex vector space $\mathfrak z$ equipped with a real form $\mathfrak z_\R\subset \mathfrak z$ and an inner product on $\mathfrak z_\R$ \cite[Thm.\,4.10]{Ai+Lin:Unitary-structures-of-VOAs}.
It is those features that we have distilled into Definition~\ref{def: g-chir WZW}:\medskip

\emph{
A general chiral WZW model is a unitary rational VOA, which is a simple current extension of a VOA that is generated in degree one.
}\medskip

The additional requirement that the VOA should be rational corresponds to $G$ being compact.
Finally, as explained in the introduction, the sub-VOA generated in degree one needs to be part of the data, as otherwise one cannot always recover the Lie algebra $\g$ from the VOA.

\section{$H^4(BG)$ for connected Lie groups}\label{sec : H^4}

Before embarking in the proof of Theorem \ref{thm: Classific WZW}, it will be important to acquire a good understanding of the group $H^4(BG,\Z)$, in which the levels live.

\begin{theorem}\label{prop: H^4(BG)}
Let $G$ be a connected Lie group and let $\g$ be its complexified Lie algebra.
Then there is a natural bijection between $H^4(BG,\Z)$ and the group of $G$-invariant symmetric bilinear forms $\langle\,\,,\,\rangle:\g\otimes \g\to \C$ that obey
\[
\tfrac12\langle X,X\rangle\in \Z
\quad\,\, \text{for every}
\quad\,\, 
X\in\g\,\;\text{s.t.}\; \exp(2\pi i X)=e.
\]
Letting $T$ be a maximal torus of $G$ and $W$ the Weyl group, we also have
$H^4(BG,\Z)=H^4(BT,Z)^W$.
\end{theorem}

We emphasize that the above result is very special to $H^4$.
For a general compact Lie group, the higher integral cohomology of $BG$ is full of torsion and admits no easy description, even when $G$ is simply connected \cite{MR593242}.

\begin{proof}[Proof of Theorem \ref{prop: H^4(BG)}:]
Let
\begin{equation}\label{eq: finite cover of G}
\tilde G\,=\,G^{ss}\times T_0\,\to\, G
\end{equation}
be a finite cover of $G$ by the product of a semi-simple simply connected Lie group $G^{ss}$ and a torus $T_0$.
Let $\tilde T$ be the preimage in $\tilde G$ of the maximal torus $T\subset G$.
Let $\mathfrak t$ be the complexified Lie algebra of $T$ (and of $\tilde T$).
Let $\Lambda=\pi_1(\tilde T)$ be the coroot lattice of $\tilde T$, identified with the set of elements $X\in \mathfrak t$ such that $\exp(2\pi i X)$ is trivial in $\tilde T$.
Finally, let $\Lambda^*:=\hom(\Lambda,\Z)$ be the weight lattice of $\tilde T$.

The cohomology ring of $B\tilde T$ is the symmetric algebra on $\Lambda^*$, and the map
\[
\mathrm{Sym}^2(\Lambda^*)\to \hom (\Lambda\otimes \Lambda,\Z):\,\,
\omega_1\!\cdot\!\!\;\!\; \omega_2\mapsto \big(\lambda_1\otimes\lambda_2\mapsto \omega_1(\lambda_1)\omega_2(\lambda_2)+\omega_1(\lambda_2)\omega_2(\lambda_1)\big)
\]
identifies $H^4 (B\tilde T, \Z)=\mathrm{Sym}^2 (\Lambda^*)$ with the group of even symmetric bilinear forms on $\Lambda$,
equivalently, symmetric bilinear forms $\langle\,\,,\,\rangle:\Lambda\otimes \Lambda\to \Z$ such that 
$\tfrac12\langle X,X\rangle\in \Z$ for every $X\in \Lambda$.
By \cite[(1.7.4)]{MR1441006}, the image of $H^4 (B\tilde G, \Z)$ in $H^4 (B\tilde T, \Z)$ is exactly its 
Weyl invariant part.
Similarly, $H^2 (B\tilde G, \Z)$ is the Weyl invariant part of $H^2 (B\tilde T, \Z)=\Lambda^*$, which is just the weight lattice of $T_0$.

Let $Z=\ker(\tilde G\to G)$ be the kernel of the projection map and 
let $K(Z,2)$ be the corresponding second Eilenberg-MacLane space.
The Serre spectral sequence for the fibration $B\tilde G\to BG \to K(Z,2)$ looks as follows:
\begin{equation}\label{eq: SS1}
\tikzmath{
\node at (0,0) {$\Z$};
\node at (1.5,0) {$\scriptstyle0$};
\node at (3,0) {$\scriptstyle0$};
\node (A2) at (5,0) {$Z^*$};
\node at (6.5,0) {$\scriptstyle0$};
\node[scale=.9] (B3) at (8.8,0) {$H^5(K(Z,2))$};
\node at (0,1) {$\scriptstyle0$};
\node at (1.5,1) {$\scriptstyle0$};
\node at (3,1) {$\scriptstyle0$};
\node at (5,1) {$\scriptstyle0$};
\node[scale=.9] (A1) at (0,2) {$H^2(BT_0)$};
\node at (1.5,2) {$\scriptstyle0$};
\node at (3,2) {$\scriptstyle0$};
\node[scale=.9] (B2) at (5,2) {$H^2(BT_0)\otimes Z^*$};
\node at (0,3) {$\scriptstyle0$};
\node at (1.5,3) {$\scriptstyle0$};
\node at (3,3) {$\scriptstyle0$};
\node[scale=.9] (B1) at (0,4) {$H^4(B\tilde T)^W$};
\draw[->] (A1) -- (A2);
\draw[->] (B1) -- (B2);
\draw[->] (B1) -- (B3);
}
\end{equation}
\[
H^4(BG)=\ker\big(d_5:\ker(d_3:H^4B\tilde G\to H^2BT_0\otimes Z^*)\to H^5K(Z,2)\big).\smallskip
\]
Here, $Z^*=H^3(K(Z,2),\Z)=\hom(Z,U(1))$ and $H^5(K(Z,2),\Z)=H^4(K(Z,2),U(1))$ is the group of $U(1)$-valued quadratic forms on $Z$ \cite[Thm 26.1]{MR0065162}.

We wish to compare the above spectral sequence 
with the one associated to the fibration $B\tilde T\to BT \to K(Z,2)$:

\begin{equation}\label{eq: SS2}
\tikzmath{
\node at (0,0) {$\Z$};
\node at (1.5,0) {$\scriptstyle0$};
\node at (3,0) {$\scriptstyle0$};
\node (A2) at (5,0) {$Z^*$};
\node at (6.5,0) {$\scriptstyle0$};
\node[scale=.9] (B3) at (8.8,0) {$H^5(K(Z,2))$};
\node at (0,1) {$\scriptstyle0$};
\node at (1.5,1) {$\scriptstyle0$};
\node at (3,1) {$\scriptstyle0$};
\node at (5,1) {$\scriptstyle0$};
\node[scale=.9] (A1) at (0,2) {$H^2(B\tilde T)$};
\node at (1.5,2) {$\scriptstyle0$};
\node at (3,2) {$\scriptstyle0$};
\node[scale=.9] (B2) at (5,2) {$H^2(B\tilde T)\otimes Z^*$};
\node at (0,3) {$\scriptstyle0$};
\node at (1.5,3) {$\scriptstyle0$};
\node at (3,3) {$\scriptstyle0$};
\node[scale=.9] (B1) at (0,4) {$H^4(B\tilde T)$};
\draw[->] (A1) -- (A2);
\draw[->] (B1) -- (B2);
\draw[->] (B1) -- (B3);
}
\end{equation}
\[
H^4(BT)=\ker\big(d_5:\ker(d_3:H^4B\tilde T\to H^2B\tilde T\otimes Z^*)\to H^5K(Z,2)\big).\smallskip
\]
The map from the $E_2$ page of the spectral sequence \eqref{eq: SS1} to the $E_2$ page of the spectral sequence \eqref{eq: SS2} is injective in all the bidegrees drawn. 
It follows that
\[
H^4(BG)=H^4(B\tilde G)\cap H^4(BT)=H^4(BT)^W,
\]
where the intersection takes place inside $H^4(B\tilde T)$.

To finish the argument, we recall that 
the set of vectors $X\in \g$ which are conjugate to elements of $\mathfrak t$ is Zariski dense.
So there is a bijection between $G$-invariant (equivalently $\tilde G$-invariant) symmetric bilinear forms on $\g$ and $W$-invariant symmetric bilinear forms on $\mathfrak t$:
\[
\begin{split}
H^4(BG)=H^4(BT)^W&=\,
\raisebox{-2.64mm}{\parbox{9.5cm}{
$W$-invariant symmetric bilinear forms $\langle\,\,,\,\rangle:\mathfrak t\otimes \mathfrak t\to \C$ s.t. $\tfrac12\langle X,X\rangle\in \Z$ for every $X\in \mathfrak t$ with $\exp(2\pi i X)=e$}}\\
&=\,
\raisebox{-2.64mm}{\parbox{9.5cm}{
$G$-invariant symmetric bilinear forms $\langle\,\,,\,\rangle:\g\otimes \g\to \C$ s.t. $\tfrac12\langle X,X\rangle\in \Z$ for every $X\in\g$ with $\exp(2\pi i X)=e$}}
\end{split}
\]
where the exponential map takes its values in $T$, respectively in $G$.
\end{proof}

\begin{remark*}
Given a homomorphism $G_1\to G_2$ of connected Lie groups, the induced map in cohomology $H^4(BG_2,\Z)\to H^4(BG_1,\Z)$ corresponds to restriction of bilinear forms.
This is clear when the groups are tori, and 
the general case follows since $H^4(BG,\Z)$ is always a subgroup of the corresponding cohomology group $H^4(BT,\Z)$ for the maximal torus.
\end{remark*}

Under the isomorphism provided by the above proposition, the subset $H^4_+(BG,\Z)$ of positive levels corresponds to those 
$G$-invariant symmetric bilinear forms $\langle\,\,,\,\rangle:\g\otimes \g\to \C$ which satisfy the property that $\tfrac12\langle X,X\rangle\in \Z_{> 0}$ for every non-zero $X\in\g$ with $\exp(2\pi i X)=e$.
The above description is slightly more precise than the one given in Section~\ref{intro2}, because we had not specified our normalization of the Chern-Weil homomorphism, and so the exact meaning of $H^4_+(BG,\Z)$ was actually left ambiguous.
The archetypal example of a positive level is given by the bilinear form $$\langle A,B\rangle :=tr(AB)\quad\,\,\, \text{on}\quad\mathfrak{sl}(n).$$
It corresponds to the positive generator of $H^4(BSU(n),\Z)$.

\section{Representations of affine and Heisenberg VOAs}\label{sec: PER}

In this section, we recall the well-known classification of representations of the affine VOAs $L_{\g,k}$ and of the Heisenberg VOAs $M_{\mathfrak z}$ (see \cite[\S3]{MR1159433}\cite[\S2]{MR1822111} and references therein, along with \cite{MR3119224} for a discussion of unitarity).

\subsection{Affine VOAs}\label{sec: aff VOA reps}
Let $\g$ is a simple Lie algebra over $\C$, with Cartan subalgebra $\mathfrak h$.
Let
$\Lambda_\text{root}\subset \Lambda_\text{weight} \subset \mathfrak h^*$ and
$\Lambda_\text{coroot} \subset \Lambda_\text{coweight} \subset \mathfrak h$
be the root, weight, coroot, and coweight lattices
(the root and weight lattices are dual to the coweight and coroot lattices, respectively).
Let $\alpha_1,\ldots,\alpha_n$ be the simple roots, let $\alpha_\mathrm{max}$ be the highest root, and let 
\[
\A=\A(\g):=\{X\in \mathfrak h\,|\,\alpha_i(X)\ge 0\text{ and }\alpha_\mathrm{max}(X)\le 1\}
\]
be the Weyl alcove.
When $\g=\g_1\oplus\ldots\oplus\g_n$ is a direct sum of simple Lie algebras,
we write $\A=\A(\g):=\A(\g_1)\times\ldots\times\A(\g_n)$ 
for the product of the corresponding alcoves.

It is well known that the irreducible positive energy representations of the affine VOA $L_{\g,k}$ are classified by the orbits
of the lattice $k^{-1}(\Lambda_\text{weight})$ under the action of the affine Weyl group $\widehat W:=W\ltimes\Lambda_\text{coroot}$,
and that these representations are all unitary \cite[Thm.\,3.1.3]{MR1159433}\cite[Thm.\,4.8]{MR3119224}.

Let $G$ be the compact simply connected Lie group associated to $\g$.
The level $k\in H^4_+(BG,\Z)$ can be treated as a number when $\g$ is simple.
But it better to think of it as the linear map $k:\mathfrak h\to \mathfrak h^*$ (coming from a bilinear form of $\mathfrak h$), which we 
use to pull back the weight lattice $\Lambda_\text{weight}\subset \mathfrak h^*$ to a lattice $k^{-1}(\Lambda_\text{weight})\subset \mathfrak h$.
Every orbit of $\widehat W$ intersects the Weyl alcove in exactly one point.
The irreps of $L_{\g,k}$ are therefore also classified by
\begin{equation}\label{slfjvksdjfvkja vk}
\A_k=\A_k(\g):=k^{-1}(\Lambda_\text{weight})\cap \A\subset\A.
\end{equation}
We illustrate the sets $\A_k$ with some rank two examples:
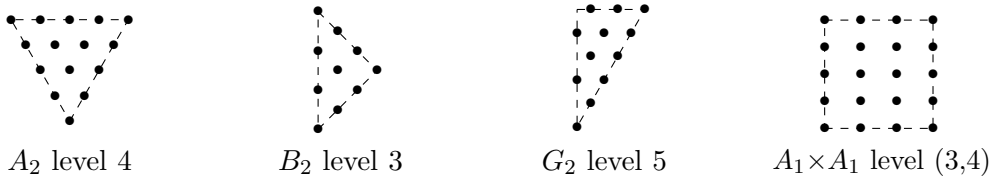
\[\quad
\def\nb{node {$\scriptstyle \bullet$}}
\tikzmath[scale=.45]{
\draw[dashed] (-90:2) -- (30:2) -- (150:2) -- cycle;
\path (-90:2) \nb (30:2) \nb (150:2) \nb (-90:.5) \nb (30:.5) \nb (150:.5) \nb;
\path (90:1) \nb +(-30:1.5) \nb +(-150:1.5) \nb;
\path (-150:1)\nb +(90:1.5) \nb +(-30:1.5) \nb;
\path (-30:1) \nb +(90:1.5) \nb +(-150:1.5) \nb;
\node at (0,-3.2) {$A_2$ level 4};
\node at (8,-3.2) {$B_2$ level 3};
\node at (15.8,-3.2) {$G_2$ level 5};
\node[scale=.97] at (24,-3.2) {$A_1{\times} A_1$ level (3,4)};
\pgftransformscale{.58}
\pgftransformxshift{360}
\pgftransformyshift{-110}
\draw[dashed] (0,0) -- (3,3) -- (0,6) -- cycle;
\path (0,0) \nb ++(1,1) \nb ++(1,1) \nb ++(1,1) \nb;
\path (0,2) \nb ++(1,1) \nb ++(1,1) \nb;
\path (0,4) \nb ++(1,1) \nb;
\path (0,6) \nb;
\pgftransformscale{1.38}
\pgftransformxshift{272}
\pgftransformyshift{126}
\draw[dashed] (0,0) -- ++(2.5,0) -- ++(-120:5) -- cycle;
\path (.5,0) \nb -- ++(1,0) \nb -- ++(1,0) \nb
 -- ++(-120:1) \nb -- ++(-1,0) \nb -- ++(-1,0) \nb
 -- ++(-60:1) \nb -- ++(1,0) \nb
  -- ++(-120:1) \nb -- ++(-1,0) \nb -- ++(-60:1) \nb  -- ++(-120:1) \nb;
\pgftransformxshift{260}
\pgftransformyshift{-125}
\draw[dashed] (0,0) rectangle (4,4);
\path (0,0) \nb -- ++(1.333,0) \nb -- ++(1.333,0) \nb -- ++(1.334,0) \nb;
\path (0,1) \nb -- ++(1.333,0) \nb -- ++(1.333,0) \nb -- ++(1.334,0) \nb;
\path (0,2) \nb -- ++(1.333,0) \nb -- ++(1.333,0) \nb -- ++(1.334,0) \nb;
\path (0,3) \nb -- ++(1.333,0) \nb -- ++(1.333,0) \nb -- ++(1.334,0) \nb;
\path (0,4) \nb -- ++(1.333,0) \nb -- ++(1.333,0) \nb -- ++(1.334,0) \nb;
}
\]
Let $\A^\times\subset \A_k$ be the subset of \emph{sharp corners} of the Weyl alcove \cite{MR1880320}:
\begin{equation}\label{eq: From Sawin}
\A^\times\,:=\,\Lambda_\text{coweight}\cap\A_k=\,\Lambda_\text{coweight}\cap \A\,\cong\, \Lambda_\text{coweight}/\Lambda_\text{coroot}.
\end{equation}
It can also be described as the orbit of the origin $0\in\mathsf A$ under the isometry group of the Weyl alcove (automorphism group of the extended Dynkin diagram),
and is canonically isomorphic to the center of $G$ \cite[Thm.\,3]{MR1880320} (see also \cite[\S4, Prop.\,8]{MR682756}).

The following result was proved in \cite[Prop.\,2.20]{MR1822111} by elementary means.
We provide here a less elementary, but possibly more informative proof:

\begin{proposition}
The elements of $\A^\times\subset \A_k$ correspond to invertible $L_{\g,k}$-representations.
\end{proposition}

\begin{proof}
By the combined work of 
\cite{MR1104840, MR1186962, MR1239506, MR1239507, MR1276910, MR2105507, MR1384612, MR3053762} (see also \cite{MO:178113} for a discussion of the few exceptional cases not covered by Kazhdan and Lusztig, including $E_8$ level $2$), the representation category of the VOA $L_{\g,k}$ is known to agree with that of the corresponding quantum group at root of unity.
The fusion rules are described by the quantum Racah formula \cite[\S5]{MR2286123} (see \cite[Remark~4]{MR2286123} for a very beautiful algorithm that allows one to compute all the rank two fusion multiplicities by hand).

Using the quantum Racah formula, one easily verifies that the representations $L_{\g,k}(\lambda)$ corresponding to elements $\lambda\in \A^\times$ are invertible with respect to fusion.
Indeed, the algorithm is invariant under $\mathrm{Isom}(\mathsf A)$.
So if $\mu, \nu \in \A_k$ are in the same $\mathrm{isom}(\mathsf A)$ orbit, then
$L_{\g,k}(\mu)$ is invertible if and only if $L_{\g,k}(\nu)$ is. Now, use that $L_{\g,k}(0)=L_{\g,k}$ is invertible.
\end{proof}

It was proven by Fuchs \cite{MR1096120} that for all simple Lie algebras and levels \emph{except $E_8$ at level $2$}, the subset $\A^\times\subset \A_k$ exhausts all invertible $L_{\g,k}$-modules.
In the case of $E_8$ level $2$, there is only one sharp corner, but the group of invertible modules has order two.
\medskip

\subsection{Heisenberg VOAs}\label{sec:Heis reps}
Let $\mathfrak z$ be a complex vector space equipped with a real form $\mathfrak z_\R$ and a positive definite inner product on the latter.
The representation theory of the Heisenberg VOAs $M_{\mathfrak z}$ associated to $\mathfrak z$ is very easy (e.g. \cite[Prop.\,2.17]{MR1822111}).
The irreducible representations of $M_{\mathfrak z}$ are all invertible and they
are classified by the points of $\mathfrak z$.
Among those, the unitary ones correspond to the points of $\mathfrak z_\R$.\footnote{We note that \cite[Prop.\,4.10]{MR3119224} only lists a subset of the unitary irreducible $M_{\mathfrak z}$-modules.}

\section{Simple current extensions}\label{Sec: Simple current extensions}

Simple current extensions were introduced in \cite{MR1030445} in the physics literature.

Given a vertex operator algebra $W$, we let $\mathrm{Rep}^\times(W)$ denote by its group of isomorphism classes of invertible modules.
For $\lambda\in\mathrm{Rep}^\times(W)$, we write $M_\lambda$ for a representative of the isomorphism class.
Let $$\pi\subset\mathrm{Rep}^\times(W)$$ be a subgroup.
A simple current extensions of $W$ by $\pi$ is an inclusion of VOAs $W\subset V$ (sending the conformal vector of $W$ to the conformal vector of $V$), such that $V\cong\bigoplus_{\lambda\in \pi}M_\lambda$ as a $W$-module, and such that the VOA structure on $V$ is compatible with the $\pi$-grading.
We write $$V=\pi\ltimes W$$ for the simple current extension.

The following proposition was shown in \cite{arXiv:1408.5215} (see \cite[Thm.\,3.3]{MR3119224} for a special case), without reference to unitarity, modulo a certain issue that the author was not able to solve.
In the more recent paper \cite{arXiv:1511.08754},
the issue was resolved for unitary VOAs, and non-unitary counter-examples were presented.\footnote{We note that a similar result in the world of chiral conformal nets was proved in \cite[Lem.\,2.1]{MR2263720} (chiral conformal nets are conjectured to be equivalent to unitary VOAs \cite{arXiv:1503.01260}).}

\begin{proposition}[\cite{arXiv:1408.5215},\cite{arXiv:1511.08754}]\label{prop/conj}
Let $W$ be a unitary vertex operator algebra, and let $\pi\subset\mathrm{Rep}^\times(W)$ be a subgroup of its invertible unitary modules.

Then the simple current extension $\pi\ltimes W$ exists if and only if the modules $M_\lambda$ corresponding to the elements $\lambda\in \pi$ have integral $L_0$-eigenvalues.
In that case, the VOA $\pi\ltimes W$ is unique up to isomorphism.
\end{proposition}

\begin{conjecture}\label{conj: unitary simple current ext}
For $W$ and $\pi$ as above, the simple current extension $\pi\ltimes W$ is unitary.
\end{conjecture}

Clearly, if $\bigoplus_{\lambda\in \pi}M_\lambda$ admits a VOA structure, then the $L_0$-eigenvalues of each $M_\lambda$ must be integral.
The difficult part is to show, provided the latter condition is satisfied, that $\bigoplus_{\lambda\in \pi}M_a$ admits a VOA structure.
We provide a proof of Proposition~\ref{prop/conj} when $W$ is rational in the sense of \cite[Appendix]{MR3339173}, by relying on the classification of extensions established in \cite{MR3339173}.

\begin{proof}
Let $W$ be a unitary VOA which satisfies the conditions named in \cite[Appendix]{MR3339173} (necessary for the results in that paper to apply), and let $\pi\subset\mathrm{Rep}^\times(W)$ be a group of invertible modules.

Let $\mathcal{C}_\pi\subset \mathrm{Rep}(W)$ be the subcategory spanned by $\pi$, i.e.,
the full subcategory whose objects are isomorphic to direct sums of elements in $\pi$.
A braided tensor category all of whose objects are invertible is entirely determined by its fusion rules and by the self-braiding of its simple objects \cite[Prop 2.14]{MR2076134}.
The self-braidings are equal to the conformal spins $\theta_\lambda=e^{2\pi i L_0}|_{M_\lambda}$ by \cite[(1.1) \& (1.2)]{arXiv:1511.08754} (this is where unitarity gets used).
By assumption, those numbers are all $1$ for $\lambda\in\pi$, and so $\mathcal{C}_\pi$ is equivalent as a braided tensor category to the category of $\pi$-graded vector spaces, with its trivial braiding:
\[
\mathcal{C}_\pi\,\simeq\, \mathrm{Vec}[\pi].
\]
The graded vector space $\bigoplus_{\lambda\in \pi}\C_\lambda$ admits a unique commutative graded algebra structure in $\mathrm{Vec}[\pi]$.
It follows that $\bigoplus_{\lambda\in \pi}M_\lambda\in\mathcal{C}_\pi\subset\mathrm{Rep}(W)$ also has a unique commutative algebra structure.
By the main result of \cite{MR3339173}, we conclude that the $\bigoplus_{\lambda\in \pi}M_\lambda$ admits a unique VOA structure extending that of $W$.
\end{proof}

We finish this section by analyzing our main example of interest: $W=L_{\g,k}\otimes M_{\mathfrak z}$.
We assume that $L_{\g,k}\otimes M_{\mathfrak z}$ is equipped with a unitary structure,
coming from a real form of $\g$ and of $\mathfrak z$.
We also assume for simplicity that
$\pi$ does not involve the non-trivial invertible $L_{E_8,2}$-module, so that
$\pi\subset \A^\times\times\mathfrak z_\R$.
Let $L:= \pi \cap \mathfrak z_\R$, and let us
furthermore assume that $\mathrm{rk}(L)=\dim(\mathfrak z)$, so that $V_L := L\ltimes M_{\mathfrak z}$ is a lattice VOA.
The simple current extension $L\ltimes (L_{\g,k}\otimes M_{\mathfrak z})$ is isomorphic to the tensor product $L_{\g,k}\otimes V_L$, and is unique by \cite[Prop.\,3.22]{MR1822111}.
The existence and uniqueness of
\[
\pi\ltimes (L_{\g,k}\otimes M_{\mathfrak z})=(\pi/L)\ltimes \big(L\ltimes (L_{\g,k}\otimes M_{\mathfrak z})\big)
=(\pi/L)\ltimes (L_{\g,k}\otimes V_L)
\]
then follows from our earlier proof, since $L_{\g,k}\otimes V_L$ is rational
in the sense of~\cite{MR3339173}.
This provides an independent proof of Proposition \ref{prop/conj} in our case of interest.

We remark that simple current extensions of $L_{\g,k}\otimes M_{\mathfrak z}$
were the main object of study of \cite{MR1822111}.
Our results are small variations of the ones in that paper.

\section{The minimal energy}\label{sec: min en}

The \emph{minimal energy} of a representation is the smallest eigenvalue of the energy operator $L_0$.
Given a representation $M_\lambda$ of some unitary VOA, we write $h_\lambda\in\R$ for its minimal energy.
It is related to the conformal spin by the formula $\theta_\lambda=\exp(2\pi i h_\lambda)$.

In view of Proposition \ref{prop/conj}, it is important to compute the minimal energies of the invertible unitary representations of $L_{\g,k}\otimes M_{\mathfrak z}$
and to determine whether they are integers.

\subsection{Affine VOAs}

Recall from Section \ref{sec: aff VOA reps} that the irreducible representations of the affine VOA $L_{\g,k}$ are classified by the finite set $\A_k\subset \mathsf A$.
For $\lambda\in \A_k$, we denote by $h_\lambda\in\R$ the minimal energy of the corresponding representation $M_\lambda=L_{\g,k}(\lambda)$.
When $\g$ is simple, it is given by the well known formula \cite[(12.8.11)]{MR1104219}
\begin{equation}\label{eq: h -- ss case}
h_\lambda\,=\,\frac{\|\lambda+\rho\|^2-\|\rho\|^2}{2(k+g^\vee)}.
\end{equation}
Here, $\rho$ is the sum of the fundamental weights (also the half-sum of the positive roots), $g^\vee=\langle\rho,\alpha_\mathrm{max}\rangle+1$ is the dual Coxeter number ($\alpha_\mathrm{max}$ is the highest root),
the square norm is taken with respect to the basic inner product on $\mathfrak h^*$ 
(the one for which long roots have square-length $2$),
and $k\in H^4(BG,\Z)=\Z$ is treated as a number.
Note that, for the purpose of equation \eqref{eq: h -- ss case}, $\lambda\in \mathsf A_k\subset \mathfrak h$ is now viewed an element of the weight lattice, via the isomorphism $\mathfrak h\to \mathfrak h^*$ provided by $k$.

When $\lambda$ corresponds to an invertible representation (and $(\g,k)$ is not $E_8$ level $2$), the above formula simplifies greatly:

\begin{proposition}\label{min en computation}
Let $\mathfrak g$ be a simple Lie algebra.
Let $\omega$ be an element of $\A^\times\subset \mathfrak h$, and
let $\lambda:=k\omega\in \mathfrak h^*$ be the corresponding weight.
Let $\langle\,\,,\,\rangle_k:\g\otimes\g\to\C$ be the symmetric bilinear form associated to $k\in H^4_+(BG,\Z)$ under Theorem \ref{prop: H^4(BG)}.
Then we have
\begin{equation}\label{eq: h_lambda = ...}
h_\lambda\,=\,\tfrac12\langle\omega,\omega\rangle_k.
\end{equation}
\end{proposition}

\begin{proof}
Let $\langle\,\,,\,\rangle$ denote the basic inner product on $\mathfrak h^*$.
We identify $\mathfrak h$ with its dual $\mathfrak h^*$ via the basic inner product, so as to have $\langle\,\,,\,\rangle_k$  on $\mathfrak h$ correspond to $k\langle\,\,,\,\rangle$ on $\mathfrak h^*$.
The equation $\lambda=k\omega$ can be interpreted in the following two equivalent ways: (1) viewing $k$ as bilinear form, $\lambda$ is the image of $\omega$ under the induced map $k:\mathfrak h\to \mathfrak h^*$;
(2) treating $k$ as a mere number and using the basic inner product to identify $\omega\in \mathfrak h$ with an element of $\mathfrak h^*$, again denoted $\omega$, then 
$\lambda$  is $k$ times $\omega$.
In what follows, we will use the second interpretation.
We will show that
\begin{equation}\label{eq: minimal energy formula}
h_\lambda\,\,=\,\,\frac{\| k\omega+\rho\|^2-\|\rho\|^2}{2(k+g^\vee)}
\,\,=\,\,
\frac{k\langle\omega,\rho\rangle}{g^\vee}
\,\,=\,\,k\!\;\!\;{\cdot}\!\; \tfrac 12\langle\omega,\omega\rangle.
\end{equation}
The first equality is just \eqref{eq: h -- ss case}.
By Lemma \ref{lem: isom(A)}, the numerator $\| k\omega+\rho\|^2-\|\rho\|^2$ vanishes when $k=-g^\vee$. The function $k\mapsto \frac{\| k\omega+\rho\|^2-\|\rho\|^2}{2(k+g^\vee)}$ is therefore linear. One easily checks that it vanishes at zero and that its derivative at zero is $\frac{\langle\omega,\rho\rangle}{g^\vee}$.
This establishes the second equality in \eqref{eq: minimal energy formula}.
Finally, by the same Lemma \ref{lem: isom(A)}, the point $\frac \rho {g^\vee}$ is equidistant to $0$ and to $\omega$.
It is therefore on the bisecting hyperplane of the segment $[0,\omega]$, and so $\langle\omega,\frac\rho{g^\vee}\rangle=\langle\omega,\frac\omega2\rangle=\tfrac12\langle\omega,\omega\rangle$.
\end{proof}

\begin{lemma}\label{lem: isom(A)}
Let $\omega$ be as above (a vertex of $\mathsf A$ in the $\mathrm{isom}(\mathsf A)$-orbit of $0$), and let $g^\vee$ be the dual Coxeter number. Then $\| \rho - g^\vee\omega\| = \|\rho\|$.
\end{lemma}

\begin{proof}
The dual Coxeter number $g^\vee$ is such that
$\rho\in\Lambda_\text{weight}$ is the unique weight in the interior of $g^\vee \mathsf A$.
The vector $\rho$ is therefore fixed under the action of $\mathrm{isom}(g^\vee \mathsf A)$.
The vertices $0$ and $g^\vee\omega$ are in the same orbit of the isometry group of $g^\vee \mathsf A$.
They are therefore equidistant to $\rho$.
\end{proof}

\subsection{Heisenberg VOAs}

As mentioned in Section \ref{sec:Heis reps}, the unitary representations of the Heisenberg VOA $M_{\mathfrak z}$ are parametrized by the real part $\mathfrak z_\R$ of $\mathfrak z$.
We write $M_{\mathfrak z}(\lambda)$ for the representation of $M_{\mathfrak z}$ corresponding to $\lambda\in \mathfrak z_\R$.
Its minimal energy is given by the well-known formula
\begin{equation}\label{eq: min energy Heis}
h_\lambda\,=\,\tfrac12\langle\lambda,\lambda\rangle_{\mathfrak z}.
\end{equation}
The striking similarity between the formulas
(\ref{eq: h_lambda = ...}) and (\ref{eq: min energy Heis}) can be explained by the fact that
the same twisting construction \cite[\S3]{MR1475118} can be used to construct $L_{\g,k}(\lambda)$ from $L_{\g,k}$ and $M_{\mathfrak z}(\lambda)$ from $M_{\mathfrak z}$.
In both cases, the minimal energy of the resulting representation is given by formula \cite[(2.41)]{MR1822111}.
This provides a proof of equation (\ref{eq: min energy Heis}), and an alternative proof of Proposition~\ref{min en computation}.

\section{The classification of chiral WZW models}

Let $G$ be a compact connected Lie group.
Pick a finite cover $G^{ss}\times T_0$ of $G$ which is the product of a semi-simple simply connected Lie group $G^{ss}$ and a torus $T_0$.
Let $\tilde G$ and $\tilde T_0$ be the universal covers of $G$ and of $T_0$,
so that
\[
\tilde G\,\,\cong\,\, G^{ss}\times \tilde T_0.
\]
We identify the fundamental group $\pi:=\pi_1(G)$ with the kernel of the projection $\tilde G \to G$.
It is a subgroup of the center $\pi\subset Z(\tilde G)$, and satisfies $\mathrm{rk}(\pi\cap \tilde T_0)=\dim(T_0)$.

Conversely, if we start from a pair $(\tilde G,\pi)$
with $\pi\subset Z(\tilde G)$, if we assume that
$\tilde G=G^{ss}\times \tilde T_0$ is the product of a compact simply connected Lie group $G^{ss}$ by a real vector space $\tilde T_0$, and if we assume that $\pi\subset Z(\tilde G)$ is discrete and satisfies $\mathrm{rk}(\pi\cap \tilde T_0)=\dim(T_0)$,
then the quotient group $G:=\tilde G/\pi$ is compact.
These two constructions are each other's inverses, and establish a bijective correspondence between compact connected Lie groups $G$ and pairs $(\tilde G,\pi)$ as above.

Let us agree that a \emph{real structure} on a complex Lie algebra $\g$ consists of a real vector space $\g_\R\subset \g$ whose complexification is 
$\g$, and such that $i\g_\R$ is closed under the Lie bracket of $\g$.
A real structure is of \emph{compact type} if there exists a compact Lie group whose Lie algebra is $i\g_\R$.
From the above discussion, we get a bijection
\begin{equation}\label{eq: G bijection}
\begin{split}
\big\{G\,\,\big|\,G\,\, &\text{is a compact connected Lie group}\big\}
\\
&\hspace{1.6cm}\updownarrow
\\
\big\{(\g,\,\pi\subset Z(\tilde G)\,\big|\,\,&
\g\,\, \text{is a Lie algebra equipped with a real}
\\
&\text{structure of compact type; $\tilde G$ is the simply}
\\
&\text{connected Lie group associated to $i\g_\R$;}
\\
&\pi\,\,\text{satisfies}\,\,\mathrm{rk}\big(\pi\cap \exp(i\mathfrak z_\R)\big)=\dim(\mathfrak z),
\\
&\text{where}\,\,\mathfrak z=\mathfrak z(\g)\,\, \text{and}\,\,
\mathfrak z_\R=\mathfrak z\cap \g_\R
\big\}
\end{split}
\end{equation}

Given the extra data of a positive level $k\in H^4(BG,\Z)$, then, by Theorem \ref{prop: H^4(BG)}, the complexified Lie algebra $\g$ comes equipped with an invariant symmetric bilinear form $\langle\,\,,\,\rangle_k:\g\otimes\g\to \C$
such that
for every $X\in\g$ with $\exp(2\pi i X)=e$ in $G$, we have $\tfrac12\langle X,X\rangle_k\in \Z$.
The level $k$ is positive if $\tfrac12\langle X,X\rangle_k > 0$ for every non-zero $X$ with $\exp(2\pi i X)=e$.

Let $T^{ss}$ be a maximal torus of $G^{ss}$ and let $\mathfrak h$ be its complexified Lie algebra.
We identify the coroot lattice with the set of elements $X\in \mathfrak h$ such that $\exp(2\pi i X)$ is trivial in $G^{ss}$, and write
\[
\mathfrak h_\R:=\mathfrak h\cap \g_\R=\Lambda_{\text{coroot}}\otimes_\Z\R
\]
for the real span of the coroot lattice (or coweight lattice).
The vector space $\mathfrak h_\R$ can also be described as $i$ times the Lie algebra of $T^{ss}$.

Putting together the formulas (\ref{eq: h_lambda = ...}) and (\ref{eq: min energy Heis}) from the previous section, we obtain:
\begin{corollary}\label{Cor: min energy}
Let $\g^{ss}=\g_1\oplus\ldots\oplus\g_n$ be a semi-simple Lie algebra equipped with a real structure of compact type.
Let $G^{ss}$ be the associated simply connected compact Lie group, and let $k=(k_1,\ldots,k_n)\in H^4_+(BG^{ss},\Z)$ be a level.
Let $\mathfrak h$ be a Cartan subalgebra of $\g^{ss}$, and let 
$\langle\,\,,\,\rangle_k:\mathfrak h\otimes \mathfrak h\to \C$ be the symmetric bilinear form associated to $k$.

Let $\mathfrak z$ be a complex vector space, with 
symmetric bilinear form 
$\langle\,\,,\,\rangle_{\mathfrak z}:\mathfrak z\otimes \mathfrak z\to \C$, and let $\mathfrak z_\R$ be a real form of $\mathfrak z$ on which the bilinear form restricts to a positive definite inner product.

Let $W=L_{\g^{ss},k}\otimes M_{\mathfrak z}$ be the tensor product of the affine and Heisenberg VOAs corresponding to the above data, and let
$M_\lambda$ be an invertible unitary $W$-module,
classified by some $\lambda\in \A^\times\times \mathfrak z_\R\subset \mathfrak h\times \mathfrak z$. Then the minimal energy of $M_\lambda$ is given by
\[
h_\lambda\,=\,\tfrac12\langle\lambda,\lambda\rangle_{k\oplus \mathfrak z},
\]
where $\langle\,\,,\,\rangle_{k\oplus \mathfrak z}$
denotes the symmetric bilinear form on $\mathfrak h\oplus \mathfrak z$ which is the direct sum of the forms $\langle\,\,,\,\rangle_k$ on $\mathfrak h$
and $\langle\,\,,\,\rangle_{\mathfrak z}$ on $\mathfrak z$.
\end{corollary}

\noindent
We summarize what we have learned so far into the following commutative diagram:
\begin{equation}\label{summary}
\qquad\begin{matrix}
\begin{tikzpicture}
\node (A) at (-1.5,1.5) {$Z(\tilde G)$};
\node (B) at (-1.5,0) {$\tilde G$};
\node (C) at (2.2,1.5) {$\A^\times\times\mathfrak z_\R$};
\node (D) at (2.2,0) {$\mathfrak h_\R\times\mathfrak z_\R$};
\node (E) at (6.5,1.5) {$\Rep^\times_{\text{unitary}}(L_{\g^{ss},k}\otimes M_{\mathfrak z})$};
\node (F) at (6.5,0) {$\R$};
\node (G) at (9.5,0) {$U(1)$};
\draw[right hook->] (A) -- (B);
\draw[->] (C) --
node[above]{$\scriptstyle\exp(2\pi i\cdot)$}
node[below]{$\scriptscriptstyle\cong$} (A);
\draw[right hook->] (C) -- (D);
\draw[>->] (C) -- (E);
\draw[->] (E.south-|F) --node[right]{$\scriptstyle h$} (F);
\draw[->] (D) --node[above]{$\scriptstyle\exp(2\pi i\cdot)$} (B);
\draw[->] (D) --node[above]{$\scriptstyle\frac 1 2 \|\,\cdot\,\|^2_{k\oplus \mathfrak z}$} (F);
\draw[->] (F) --node[above]{$\scriptstyle\exp(2\pi i\cdot)$} (G);
\draw[->] (E.south) ++ (1.3,0) --node[right, yshift=4]{$\scriptstyle\theta$} (G);
\end{tikzpicture}
\end{matrix}
\end{equation}
Here, $h$ is the minimal energy and $\theta$ is the conformal spin.
The map $\A^\times\times\mathfrak z_\R\to Z(\tilde G)$ is an isomorphism by \cite[Thm.\,3]{MR1880320}, and the commutativity of the middle rectangle is the content of Corollary~\ref{Cor: min energy}.

The natural map  $\A^\times\times\mathfrak z_\R \to \Rep^\times_{\text{unitary}}(L_{\g^{ss},k}\otimes M_{\mathfrak z})$ is almost an isomorphism:
it is injective with cokernel isomorphic to $n$ copies of $\Z/2$, where $n$ is the number of occurrences of $(E_8,2)$ in the decomposition of $(\g^{ss},k)$ into simples.\medskip

Recall that a chiral WZW model (Definition~\ref{def: chir WZW}) is a pair consisting of a unitary rational VOA $V$ and a Lie algebra $\g\subset V_1$ such that $V$ is a simple current extension of the sub-VOA $W\subset V$ generated by $\g$.
The VOA $V$ is assumed to be unitary and rational ($W$ does not need to be rational),
and we are not allowed to use the non-trivial invertible $L_{E_8,2}$-module in the construction of the simple current extension.

By \cite[Thm.\,4.10]{Ai+Lin:Unitary-structures-of-VOAs}, $W$ has the following structure.
Since $W$ is generated by $\g\subset W_1$,
it is a quotient of the universal affine VOA $V_{\g,\kappa}$ associated to $\g$ and to some invariant bilinear form $\kappa:\g\otimes\g\to\C$.
Since $W$ is unitary, the Lie algebra $\g$ comes equipped with an invariant positive definite hermitian form which, together with the bilinear form $\kappa$, yields a real structure of compact type on $\g$.
The Lie algebra $\g$ is therefore a direct sum $\g^{ss}\oplus\mathfrak z$ of a semi-simple Lie algebra of compact type and an abelian Lie algebra:
\[
\g=\g^{ss}\oplus\mathfrak z=\g_1\oplus\g_2\oplus\ldots\oplus\g_n\oplus\mathfrak z.
\]
On each simple summand $\g_i$, the bilinear form $\kappa$ is a positive multiple of the basic inner product,
and $W$ is isomorphic to the tensor product 
$L_{\g^{ss},k}\otimes M_\mathfrak z$ of a number of affine VOAs and a Heisenberg VOA.

Let $\A_k$ be the finite set that parametrizes the irreducible representations of $L_{\g^{ss},k}$, and let $\A^\times\subset \A_k$ be the subset of sharp corners of the alcove.
By definition, $V$ is of the form $$V=\pi\ltimes (L_{\g^{ss},k}\otimes M_{\mathfrak z})$$
for some injective homomorphism
$\pi\to \Rep^\times_{\text{unitary}}(L_{\g^{ss},k}\otimes M_{\mathfrak z})$ that 
doesn't hit any of the stuff that involves the non-trivial invertible $L_{E_8,2}$-module.
In other words, the homomorphism factors through the image of $\A^\times\times\mathfrak z_\R$:
\[
\begin{tikzpicture}
\node (a) at (0,1) {$\pi$};
\node (b) at (1.3,0) {$\A^\times\times\mathfrak z_\R$};
\node[right] (c) at (2,1) {$\Rep^\times_{\text{unitary}}(L_{\g^{ss},k}\otimes M_{\mathfrak z})$};
\draw[dashed, ->, shorten <=1, shorten >=1] (a) -- (b);
\draw[->, shorten <=2, shorten >=1] (a) -- (c);
\draw[<-, shorten >=-2, shorten <=-2] (c.south) ++(-1.7,0) -- (b);
\end{tikzpicture}
\]
It remains to examine the condition under which $V$ is rational.
Let $L:= \pi \cap \mathfrak z_\R$.
If $\mathrm{rk}(L)=\dim(\mathfrak z)$, then $V_L := L\ltimes M_{\mathfrak z}$ is a lattice VOA.
Our VOA
is the simple current extension of the rational VOA $L_{\g,k}\otimes V_L$ by the finite group $\pi/L$
\begin{equation}\label{tilde G = ... x ...}
V=\pi\ltimes (L_{\g,k}\otimes M_{\mathfrak z})
=(\pi/L)\ltimes (L_{\g,k}\otimes V_L),
\end{equation}
and is rational by the results in \cite{MR3339173}.
If, on the contrary, $\mathrm{rk}(L)<\dim(\mathfrak z)$,
then $L\ltimes M_{\mathfrak z}$ is the product of a lattice VOA and a Heisenberg VOA, which is not rational.
The tensor product $L_{\g,k}\otimes (L\ltimes M_{\mathfrak z})$ is also not rational, and neither is its finite simple current extension
$V=(\pi/L)\ltimes (L_{\g,k}\otimes (L\ltimes M_{\mathfrak z}))$. All in all, we conclude that
\begin{equation}\label{V rational iff...}
\text{$V$ is rational}\quad \Longleftrightarrow\quad \mathrm{rk}(\pi\cap \mathfrak z_\R)=\dim(\mathfrak z).
\medskip
\end{equation}

With all the above preliminaries in place we are now ready, assuming that Conjecture~\ref{conj: unitary simple current ext} holds, to prove our main theorem:

\begin{proof}[Proof of Theorem~\ref{thm: Classific WZW}]
Consider the following sets:

\[
A_1\,:=\,\big\{(G,k)\,\big|\, G\,\, \text{is a compact connected Lie group;}\,\,k\in H^4_+(BG,\Z)\big\}
\]

\[
\begin{split}
A_2:=\big\{(G,\langle\,\,,\,\rangle:\g\otimes\g\to\C)\,\big|\,
&G\,\,\text{is a compact connected Lie group;}
\\
&\g\,\,\text{is its complexified Lie algebra};
\\
&\langle\,\,,\,\rangle\,\, \text{is $G$-invariant, positive definite on}\,\, i\,\mathrm{Lie}(G);
\\
&\tfrac12\langle X,X\rangle \in \Z\,\,\,\, \forall X\; \text{s.t.} \exp(2\pi i X)=e \text{ in } G\big\}
\end{split}
\]

\[
\begin{split}
A_3:=\big\{(\g,\pi\subset Z(\tilde G),\langle\,\,,\,\rangle:\g\otimes\g\to\C)\,\big|\,&
\g\,\, \text{is a Lie algebra equipped with a real}
\\
&\text{structure of compact type; $\tilde G$ is the simply}\quad
\\
&\text{connected Lie group associated to $i\g_\R$;}
\\
&\langle\,\,,\,\rangle\,\, \text{is $\g$-invariant, positive definite on}\,\, \g_\R;
\\
&\tfrac12\langle X,X\rangle \in \Z\,\,\,\, \forall X\; \text{s.t.} \exp(2\pi i X)\in \pi;
\\
&\mathrm{rk}\big(\pi\cap \exp(\mathfrak z(i\g_\R))\big)=\dim(Z(G))\big\}
\end{split}
\]

\[
\begin{split}
\quad A_4:=\big\{(\g=\g^{ss}\oplus\mathfrak z,\,\pi\subset\A^\times\times \mathfrak z_\R,
\,\big|\,&
\g\,\, \text{is a Lie algebra equipped with a real}
\\
\langle\,\,,\,\rangle:\g\otimes\g\to\C)\,\,\,\,
&\text{structure of compact type, written as}
\\
&\text{the sum of a semi-simple Lie algebra $\g^{ss}$}
\\
&\text{and an abelian Lie algebra $\mathfrak z$; $\A^\times$ is the set}
\\
&\text{of sharp corners of the Weyl alcove of $\mathfrak g^{ss}$;}
\\
&\langle\,\,,\,\rangle\,\, \text{is $\g$-invariant, positive definite on}\,\, \g_\R;
\\
&\text{the restriction of $\langle\,\,,\,\rangle$ to each simple}
\\
&\text{summand of $\g^{ss}$ is a positive integer}
\\
&\text{multiple of the basic inner product;}
\\
&\tfrac12\langle \lambda,\lambda\rangle \in \Z\,\,\,\, \forall \lambda\in \pi;
\\
&\mathrm{rk}(\pi\cap \mathfrak z_\R)=\dim(\mathfrak z_\R)\big\}
\end{split}
\]

\[
\begin{split}
A_5:=\big\{(\g=\g^{ss}\oplus\mathfrak z=\g_1\oplus\ldots\oplus\g_n\oplus\mathfrak z,\,
\big|\,&
\g_i\,\, \text{are simple Lie algebras, $\mathfrak z$ is an}
\\
k=(k_1,\ldots,k_n),\,\,\,\,\,&
\text{abelian Lie algebra, all with real}
\\
\langle\,\,,\,\rangle_{\mathfrak z}:\mathfrak z\otimes\mathfrak z\to\C,\,\,\,\,\,&
\text{structures of compact type;}\,\,k_i\in\mathbb N;
\\
\pi\subset \Rep^\times_{\text{unitary}}(L_{\g^{ss},k}\otimes M_{\mathfrak z}))\,\,\,\,&
\langle\,\,,\,\rangle_{\mathfrak z}\,\, \text{is positive definite on}\,\, \mathfrak z_\R;
\\
&L_{\g^{ss},k}\,\,\text{and}\,\,M_{\mathfrak z}\,\,\text{are the associated}
\\
&\text{affine and Heisenberg unitary VOAs;}
\\
&\pi\,\,\text{does not hit stuff involving the}
\\
&\text{non-trivial invertible $L_{E_8,2}$-module;}
\\
&\theta_\lambda=1\;\;\forall\lambda\in\pi;
\\
&\mathrm{rk}(\pi\cap \mathfrak z_\R)=\dim(\mathfrak z)\big\}
\end{split}
\]

\[
\begin{split}
A_6:=\big\{
(V,\g)\,\big|\,& V\,\,\text{is a rational unitary VOA;}
\\
&\g\,\,\text{is a subalgebra of the Lie algebra}\,\, V_1\subset V; 
\\
&V\,\text{is a simple current extension of the sub-VOA generated by}\, \g, 
\\
&\text{and the simple current extension is not ``$E_{8,2}$-contaminated''}
\big\}
\end{split}\medskip
\]

We will construct a sequence of bijections 
$$A_1\,\cong\, A_2\,\cong\, A_3\,\cong\, A_4\,\cong\, A_5\,\cong\, A_6.$$
The first bijection $A_1\cong A_2$ is the content of Theorem \ref{prop: H^4(BG)}.

The second bijection $A_2\cong A_3$ follows from (\ref{eq: G bijection}),
and from the observation that the condition $\exp(2\pi iX)=e$ in $G$ is equivalent to
the condition $\exp(2\pi iX)\in \pi$ in $\tilde G$.

The bijection $A_3\cong A_4$ uses the 
canonical isomorphism
$Z(\tilde G)\cong \A_k^\times\times\mathfrak z_\R$
from the top left of (\ref{summary}).
To go from $A_3$ to $A_4$, we need to check that
the restriction of $\langle\,\,,\,\rangle$ to any simple
summand $\g_i\subset\g^{ss}$ is a positive multiple
of the basic inner product on $\g_i$.
Let $\tilde G_i$ be the simply connected compact Lie group associated to $\g_i$.
By definition, the basic inner product on $\g_i$ is the one that assigns square norm $2$ to all the short coroots $\alpha_{\text{short}}\in \g_i$.
The latter satisfy
$\exp(2\pi i\alpha_{\text{short}})=e$ in $\tilde G_i$.
In particular, they satisfy $\exp(2\pi i\alpha_{\text{short}})\in\pi$.
By assumption, the number $k_i:=\tfrac12 \|\alpha_{\text{short}}\|^2$ is a positive integer, and so
we have $\langle\,\,,\,\rangle|_{\g_i\otimes \g_i}=k_i\cdot\langle\,\,,\,\rangle_{\text{basic}}$, as desired.
To
go
from $A_4$ to $A_3$, let 
$\pi\subset \A^\times\times\mathfrak z_\R$ be a subgroup, and let $\exp(2\pi i\cdot\pi)$
denote its isomorphic image in $Z(\tilde G)$.
We need to show that if $X\in\g$ satisfies $\exp(2\pi i X)\in \exp(2\pi i\cdot\pi)$, then its square norm is even.
Given such an $X$, then, by the last equality in (\ref{eq: From Sawin}), 
we can find an element $\alpha\in\Lambda_{\text{coroot}}$ such that $X+\alpha\in\pi$.
We need to show that the quantity
\[
\tfrac12\|X\|^2\,=\,\tfrac12\|X+\alpha\|^2-\tfrac12\|\alpha\|^2-\langle X,\alpha\rangle
\]
is an integer.
The first term $\tfrac12\|X+\alpha\|^2$ is an integer by assumption.
The second term $\tfrac12\|\alpha\|^2$ is in $\Z$ because $\alpha\in\Lambda_{\text{coroot}}$ and each $\langle\,\,,\,\rangle|_{\g_i\otimes \g_i}$ is an integer multiple of the basic inner product.
Finally, the last term $\langle X,\alpha\rangle$ is an integer by Lemma~\ref{coroot coweight} below.

The bijection $A_4\cong A_5$
follows from the isomorphism 
$$\A^\times\times\mathfrak z_\R \stackrel{\scriptscriptstyle \cong}\longrightarrow \Rep^\times_{\text{unitary}}(L_{\g^{ss},k}\otimes M_{\mathfrak z})\setminus\{\text{stuff involving $E_8$ level $2$}\}$$
and from the formula
$\theta_\lambda \,=\, \exp(2\pi i\cdot\tfrac12 \langle \lambda,\lambda\rangle_{k\oplus \mathfrak z})$
provided by Corollary \ref{Cor: min energy}.

Finally, the bijection $A_5\cong A_6$
follows from the equivalence $(\text{$V$ is rational}) \Leftrightarrow \mathrm{rk}(\pi\cap \mathfrak z_\R)=\dim(\mathfrak z)$ proved in (\ref{V rational iff...}), from the discussion contained in the three paragraphs preceding that equation,
and from Proposition~\ref{prop/conj}, which says that a simple current extension can be performed if and only if the conformal spins $\theta_\lambda$ are trivial.
It is for this last step that the result of Conjecture~\ref{conj: unitary simple current ext} is needed.
\end{proof}

\begin{lemma}\label{coroot coweight}
Let $\g$ be a simple Lie algebra, and let $\langle\,\,,\,\rangle$ be its basic inner product.
Then $\langle \Lambda_{\text{\rm coweight}},\Lambda_{\text{\rm coroot}}\rangle\subset \Z$.
\end{lemma}

\begin{proof}
We identify $\g$ with its dual $\g^*$ by means of the basic inner product.
Let $m\in\{1,2,3\}$ be the ratio between the square norm of the long roots and that of the short roots.
The coroot lattice $\Lambda_{\text{coroot}}$ is spanned by the long roots and by $m$ times the short roots.
In particular, $\Lambda_{\text{coroot}}\subset \Lambda_{\text{root}}$.
By dualizing, it follows that 
$\Lambda_{\text{coweight}}\subset \Lambda_{\text{weight}}$.
The lattices $\Lambda_{\text{weight}}$ and $\Lambda_{\text{coroot}}$ are dual to each other.
In particular, the pairing of a weight and a coroot is always an integer.
It follows that $\langle \Lambda_{\text{coweight}},\Lambda_{\text{coroot}}\rangle
\subset
\langle\Lambda_{\text{weight}},\Lambda_{\text{coroot}}\rangle = \Z$.
\end{proof}\medskip

To finish our analysis,
we compare the sets $A_{1,\ldots,6}$ from the above proof with the set of general chiral WZW models given in Definition \ref{def: g-chir WZW}:
\[
\begin{split}
A_7:=\big\{\;\!
V\,\;\!\big|\,\,& V\,\,\text{is a rational unitary VOA which is a simple current}
\\
&\text{extension of a sub-VOA that is generated in degree one}
\big\}
\end{split}\smallskip
\]

The composite
\begin{equation}\label{A7}
A_1\,\cong\, A_2\,\cong\, A_3\,\cong\, A_4\,\cong\, A_5\,\cong\, A_6\,\to\, A_7
\end{equation}
is neither surjective, nor injective.

As explained in the introduction,
the VOA $V$ associated to a simple simply laced gauge group at level $1$ is isomorphic to the one associated to its maximal torus \cite{MR0626704}. This shows that (\ref{A7}) is not injective.
Finally, the simple current extension $$\Z_2\ltimes (L_{E_8,2}\otimes L_{E_8,2})$$
of $L_{E_8,2}\otimes L_{E_8,2}$ by its invertible module $M_{E_8,2}\otimes M_{E_8,2}$ does not correspond to any Lie group.
The gauge group wants to be $(E_8\times E_8)/\Z_2$, but that a quotient does not make sense as the center of $E_8\times E_8$ is trivial.
This shows that (\ref{A7}) is not surjective.
We have used two copies of $E_8$ because $M_{E_8,2}$ has conformal spin $-1$, and we can only perform simple current extensions with invertible modules of conformal spin $1$.
\medskip

Given a compact connected Lie group $G$ with complexified Lie algebra $\mathfrak g$, and given a level $k\in H^4_+(BG,\mathbb Z)$, we write $V_{G,k}$ for the corresponding VOA, under the map \eqref{A7}.
It is a simple current extension
\begin{equation*}
V_{G,k}
=\pi\ltimes (L_{\g^{ss},k}\otimes M_{\mathfrak z})
\end{equation*} 
of the tensor product of an affine VOA (associated to the semi-simple part of $\g$) and a Heisenberg VOA (associated to the center of $\mathfrak g$) by the abelian group $\pi:=\pi_1(G)$.
Alternatively, it is a simple current extension
\begin{equation}\label{sskjvjsv}
V_{G,k}
=(\pi/L)\ltimes (L_{\g^{ss},k}\otimes V_L)
\end{equation} 
of the tensor product of an affine VOA and a lattice VOA by the finite abelian group $\pi/L$ (here, $L=\pi\cap \mathfrak z_{\mathbb R}$ is as in \eqref{tilde G = ... x ...}, where the intersection takes place inside $\A^\times\times\mathfrak z_\R$).
We note that the construction of $V_{G,k}$ does not rely on Conjecture~\ref{conj: unitary simple current ext}; only its unitarity does.

\section{WZW conformal nets}

In this section, we construct a chiral conformal net $\mathcal A_{G,k}$
for every compact connected Lie group $G$ and positive level $k\in H^4_+(BG,\mathbb Z)$.
We call these the \emph{chiral WZW conformal nets}.
These conformal nets appear at the starting point of our construction \cite[\S5]{arXiv:1503.06254} of the value of Chern-Simons theory on a point.
They were also briefly mentioned in \cite[Ex.\,5.14]{arXiv:1509.02509}.

In the remarkable paper \cite{arXiv:1503.01260}, a bijective correspondence was established between a certain class of unitary VOA (the so-called `strongly local' unitary VOAs) and a certain class of conformal nets.
It is natural to ask whether the VOAs $V_{G,k}$ constructed in \eqref{sskjvjsv} fall within the domain of applicability of that correspondence, and 
what the corresponding conformal nets are.
Unfortunately, even when $G$ is a torus, it is presently not known whether all lattice VOAs are strongly local \cite[Conj.\,8.17]{arXiv:1503.01260}.
Another problem for the strong locality of the chiral WZW VOAs is their unitarity.
We will therefore contend with pursuing a more modest goal:
construct a chiral conformal net $\mathcal A_{G,k}$
for every compact connected Lie group $G$ and positive level $k\in H^4_+(BG,\mathbb Z)$.
We conjecture that they correspond to the VOAs $V_{G,k}$ under the correspondence established in \cite{arXiv:1503.01260}.

Our strategy for defining $\mathcal A_{G,k}$ is to mimic the formula~\eqref{sskjvjsv}.
Let $\pi=\pi_1(G)$ be the fundamental group of $G$.
Decompose the complexified Lie algebra as a direct sum
\[
\g=\g^{ss}\oplus\mathfrak z=\g_1\oplus\ldots\oplus\g_n\oplus\mathfrak z
\]
of a semi-simple Lie algebra $\g^{ss}=\g_1\oplus\ldots\oplus\g_n$ and an abelian Lie algebra $\mathfrak z$.
The level $k\in H^4_+(BG,\mathbb Z)$ induces levels $k_i\in\N$ for every simple Lie algebra $\g_i$.
Let $G_i$ be the compact simply connected Lie group corresponding to $\g_i$,
and let $\cA_{G_i,k_i}$ be the associated loop group conformal net \cite{MR1231644,Toledano(PhD-thesis), MR1645078}.

Let $L=\pi\cap \mathfrak z_{\mathbb R}$ be as in \eqref{sskjvjsv}.
As explained in Section \ref{sec : H^4}, $k$ induces a metric on $\mathfrak z_\R$ which endows $L$ with the structure of an even integral lattice.
Let $\cA_L$ be the lattice conformal net associated to $L$, equivalently, the loop group conformal net associated to the torus $T_0:=L\otimes_{\mathbb Z} U(1)$ \cite{MR2261756, Staszkiewicz-thesis}.

The tensor product
\begin{equation}\label{tens prod conf net}
\cA_{G_1,k_1}\otimes \ldots \otimes \cA_{G_n,k_n} \otimes \cA_L
\end{equation}
is a chiral WZW conformal net for the gauge group $\tilde G:=G_1\times \ldots\times G_n\times T_0$.
It is generated by the centrally extended loop group of $\tilde G$, acting of its vacuum repre\-sentation.
Let $Z:=\pi/L$.
\begin{definition}\label{def: slfk ldnb}
The chiral WZW conformal net associated to the Lie group $G\cong \tilde G/Z$ and the level $k\in H^4_+(BG,\mathbb Z)$ is the simple current extension
\begin{equation} \label{tens prod conf net -- SimplCurrExt}
\cA_{G,k}:=Z\ltimes(\cA_{G_1,k_1}\otimes \ldots \otimes \cA_{G_n,k_n} \otimes \cA_L)
\end{equation}
of the conformal net \eqref{tens prod conf net} by the finite abelian group $Z$.
\end{definition}
Here, an extension of conformal nets $\cA\subset \cB$ \cite{MR1332979} is called a \emph{simple current extension} if the vacuum sector of $\cB$ decomposes as a direct sum $\bigoplus_{\lambda\in Z}H_\lambda$ of invertible $\cA$-modules, and the conformal net structure of $\cB$ is compatible with the $Z$-grading.

Definition \ref{def: slfk ldnb} requires some explanation which we now provide.
First of all, it is expected but presently not known\footnote{See \cite[\S15]{Wass_unpub_1990} for some results in that direction.} that the irreducible representations of the conformal net 
\eqref{tens prod conf net}
are classified by the set $\A_{k_1}(\g_1)\times\ldots\times \A_{k_n}(\g_n)\times L^*/L$, with $\A_{k_i}(\g_i)$ as in \eqref{slfjvksdjfvkja vk}.
In the special case of a weight $\lambda$ in the subset
\begin{equation}\label{tens prod conf net -- invertible weights}
\A_{k_1}^\times(\g_1)\times\ldots\times \A^\times_{k_n}(\g_n)\times L^*/L,
\end{equation}
we can construct the associated representation $H_\lambda$ of \eqref{tens prod conf net} `by hand',
and prove that it is invertible with respect to the fusion product.
This will be sufficient for our purposes. 

The construction of $H_\lambda$ goes as follows.
Identify the set \eqref{tens prod conf net -- invertible weights} with a subset of the center of $\tilde G$ in the obvious way.
Following \cite[\S3.4]{MR900587}\cite[Def.\,3.12]{MR2261756}, for every $\lambda$ in the above subset,
there is an outer automorphism $a_\lambda$ of the centrally extended loop group of $\tilde G$ given by `conjugation' by a path from the identity to $\lambda \in Z(\tilde G)$.
This extends to a (localisable) automorphism of the associated conformal net.
The representation $H_\lambda$ is obtained by twisting the vacuum representation $H_0$ by that automorphism.
Fusing with $H_\lambda$ sends a representation to the representation with same underlying Hilbert space and action precomposed by $a_\lambda$ \cite[Sec.\,IV.2]{MR1231644}
(this is visibly an invertible operation).

The following is the main result of this section:

\begin{theorem}
The simple current extension~\eqref{tens prod conf net -- SimplCurrExt} used to define the chiral WZW conformal net $\cA_{G,k}$ exists and is unique.
\end{theorem}

\begin{proof}
Let $H_\lambda$ be the conformal net representations which enter in the definition \eqref{tens prod conf net -- SimplCurrExt} of $\cA_{G,k}$,
and let $M_\lambda$ be the ($L_{\g^{ss},k}\otimes V_L$)-representations which enter in the definition \eqref{sskjvjsv} of $V_{G,k}$.
The twisting construction \cite[\S3.4]{MR900587}\cite[Def.\,3.12]{MR2261756} of $H_\lambda$ is identical to the twisting construction \cite[\S3]{MR1475118} of $M_\lambda$.
In particular, the $L_0$-eigenspaces of $H_\lambda$ are canonically isomorphic to those of $M_\lambda$ ($H_\lambda$ is the Hilbert space completion of $M_\lambda$).
The modules $M_\lambda$ have trivial conformal spin (otherwise \eqref{sskjvjsv} would not exist).
The same property therefore holds for the~$H_\lambda$.
Proposition~\ref{slgbsljg gb} then guarantees the existence and uniqueness of the simple current extension~\eqref{tens prod conf net -- SimplCurrExt}.
\end{proof}

\begin{proposition}[{\cite[Lem.\,2.1]{MR2263720}}]\label{slgbsljg gb}
Let $\cA$ be a conformal net and let $Z\subset\mathrm{Rep}^\times(\cA)$ be a finite subgroup of its group of invertible representations.

Then the simple current extension $Z\ltimes \cA$ exists if and only if $H_\lambda$ has trivial conformal spin for every $\lambda\in Z$.
In that case, the conformal net $Z\ltimes \cA$ is unique up to isomorphism.
\end{proposition}

\begin{proof}
The proof of \cite[Lem.\,2.1]{MR2263720} only shows existence. 
We present an alternative argument, along the same lines of our proof of Proposition~\ref{prop/conj}, which also addresses uniqueness.

Consider the subcategory of $\Rep(\cA)$ spanned by the objects $H_\lambda$, $\lambda \in Z$.
By \cite[Prop 2.14]{MR2076134}
and the conformal spin-statistics theorem \cite{MR1410566},
that subcategory is braided equivalent to the category $\Vect[Z]$ of $Z$-graded vector spaces (with trivial associator and trivial braiding).
The `group algebra' $\bigoplus_{\lambda\in Z} H_\lambda$ is a commutative unitary Frobenius algebra object in an evident way.
By applying the main result of \cite{MR1332979} to the latter, we get the desired simple current extension of $\cA$.
The extension is unique because the algebra structure on $\bigoplus_{\lambda\in Z} H_\lambda$ is unique.
\end{proof}

\section{Conclusion}

We have proposed a novel definition of chiral WZW models, as simple current extensions of the tensor product of affine and Heisenberg VOAs.
We have shown that, by fine-tuning the definition  (Definition~\ref{def: chir WZW}),
we can ensure that there is a bijective correspondence between chiral WZW models, and pairs $(G,k)$ consisting of a compact connected Lie group $G$ and a level $k\in H^4_+(BG,\Z)$.
The fine-tuning involves (1) remembering the sub-VOA that was used to perform the simple current extension, and (2) disallowing the usage of the non-trivial invertible module of $E_8$ level 2.

In the absence of the fine-tuning, i.e., if one doesn't remember the sub-VOA that was used to perform the simple current extension and if one is allowed to use the non-trivial invertible module of $E_8$ level 2,
then the natural map from pairs $(G,k)$ to chiral WZW models is neither injective nor surjective.
The correspondence
\[
\left.\left\{(G,k)\left|\,\,\parbox{5.5cm}{
$G$: compact connected Lie group\\
$\,\,k\in H^4_+(BG,\Z)$}\right\}
\,\,\,\stackrel{\cong}\longrightarrow\,\,\,\,
\right\{\text{Chiral WZW models}\right\}
\]
that we have established in this paper is therefore fragile:
depending on the exact definition of a chiral WZW model, it either is or isn't a bijection.
However, in all cases, the map is \emph{close} to being a bijection.\medskip

It is interesting to note that, unlike Chern--Simons theories whose gauge group can be disconnected (see \cite{MR1048699}), the gauge group of a chiral WZW model is \emph{necessarily connected}.
This raises the question of what is the chiral WZW model associated to a disconnected gauge group?
Whatever the answer turns out to be, mathematicians will probably need to enlarge the class of objects that they agree to call `chiral conformal field theories' in order to accommodate these yet-to-be-defined models.


\bibliographystyle{abbrv}
\bibliography{WZW}

\end{document}